\documentclass{llncs}

\usepackage{subfigure}
\usepackage{color}
\usepackage{graphicx}
\usepackage{xspace}
\usepackage{times}
\usepackage{gastex}
\usepackage{amsfonts}

\newcommand{\secref}[1]{Sec.~\ref{#1}}
\newcommand{\figref}[1]{Fig.~\ref{#1}}
\newcommand{\cA}{\mathcal{A}}
\newcommand{\cB}{\mathcal{B}}
\newcommand{\cL}{\mathcal{L}}

\newcommand{\cG}{\mathcal{G}}

\newcommand{\cAphi}{\mathcal{A}_{\varphi}}

\newcommand{\buchi}{B{\"u}chi\xspace}
\newcommand{\LTLt}{\ensuremath{\mathrm{LTL}_3}\xspace}
\newcommand{\LTL}{\ensuremath{\mathrm{LTL}}\xspace}
\newcommand{\LTLtools}{\LTLt-tools\xspace}

\newcommand{\ltlG}{\mathbf{G}}
\newcommand{\ltlF}{\mathbf{F}}
\newcommand{\ltlU}{\mathbf{U}}
\newcommand{\ltlX}{\mathbf{X}}

\begin{document}
 
\title{Monitorability of $\omega$-regular languages}

\author{Andreas Bauer}

\institute{National ICT Australia (NICTA)\thanks{\scriptsize NICTA is
    funded by the Australian Government as represented by the Department
    of Broadband, Communications and the Digital Economy and the
    Australian Research Council through the ICT Centre of Excellence
    program.} and\\The Australian National University}

\maketitle
\thispagestyle{empty}

\begin{abstract}
  Arguably, $\omega$-regular languages play an important r\^ole as a
  specification formalism in many approaches to systems monitoring via
  runtime verification.
  However, since their elements are infinite words, not every
  $\omega$-regular language can sensibly be monitored at runtime when
  only a finite prefix of a word, modelling the observed system
  behaviour so far, is available.
  
  The \emph{monitorability of an $\omega$-regular language}, $L$, is
  thus a property that holds, if for \emph{any} finite word $u$,
  observed so far, it is possible to add another finite word $v$, such
  that $uv$ becomes a ``finite witness'' wrt.\ $L$;
  that is, for \emph{any} infinite word $w$, we have that $uvw \in L$,
  or for \emph{any} infinite word $w$, we have that $uvw \not\in L$.
  This notion has been studied in the past by several authors, and it is
  known that the class of monitorable languages is strictly more
  expressive than, e.g., the commonly used class of so-called safety
  languages.
  But an exact categorisation of monitorable languages has, so far, been
  missing.
  Motivated by the use of linear-time temporal logic (LTL) in many
  approaches to runtime verification, this paper first determines the
  complexity of the monitorability problem when $L$ is given by an LTL
  formula.  Further, it then shows that this result, in fact, transfers
  to $\omega$-regular languages in general, i.e., whether they are given
  by an LTL formula, a nondeterministic B\"uchi automaton, or even by an
  $\omega$-regular expression.
\end{abstract}

\section{Introduction}

\label{sec:intro}

In a nutshell, the term runtime verification subsumes many techniques
that are used for monitoring systems, i.e., for checking their execution
as it is happening.  Naturally, there exists a variety of different
approaches to runtime verification.
In this article, we will focus on those which are based on the theory of
formal languages, where a so called \emph{monitor} checks whether or not
a consecutive sequence of observed system actions belongs to a formally
specified language.  For example, if the language comprises all
undesired system behaviours, then a positive outcome of this check would
normally lead to the raising of an alarm by the monitor, whereas if the
language describes a desired system behaviour, the monitor could be
switched off.

As a formalism to describe such languages, many runtime verification
approaches (cf.\
\cite{rosu-havelund-2005-jase,havelund2004,DBLP:conf/cav/dAmorimR05,tum-i0724}),
use \emph{linear-time temporal logic} (LTL \cite{Pnueli77}), whose
formulae describe sets (languages) of infinite words (or,
$\omega$-languages), meaning that the models of an LTL formula are
infinitely long sequences of symbols.  The rationale for using LTL to
describe properties of systems is that many systems for which formal
verification is required (at runtime or off-line) are critical and/or
reactive; that is, their failure would have catastrophic impact on its
users and/or the environment, and consequently one would like to make
assertions about the entire lifespan of such systems, some of which are
never switched off, unless they are physically broken and can be
replaced in a controlled manner.
%
%
A typical requirement for such systems, that can also easily be
formalised in LTL, would be ``the system must never enter a bad state.''
Although the monitor would require an infinitely long observation to
flag satisfaction of the property, it is always able to raise an alarm
after finitely many observations, simply due to the fact that a
violation of such a property can always be detected in the same instance
as the system entering the bad state.
%
%
Hence, if such a property, formalised as an LTL formula, is monitored,
one would expect the monitor to only detect violations.  Formal
languages which describe properties of this form are therefore referred
to as \emph{safety} languages or safety properties, and they have in
common that all sequences of actions that violate them are detectable
after finitely many observations.
Note that languages belonging to the complementary class of safety
properties are known as the \emph{co-safety} properties, implying that
satisfaction (rather than violation) of any such type of property is
always detectable by a monitor after finitely many observations, i.e.,
via a finite ``witness.''

Since the languages definable by LTL formulae exceed the expressiveness
of safety and co-safety languages, a natural question to ask, given an
arbitrary LTL formula, is whether or not the given formula is
monitorable at all.
This is, arguably, an interesting question in its own right, 
and ideally, we would like to know the answer prior to any attempts of
building a monitor, or starting a monitoring process based on an
unmonitorable language.
Of course, what we then need is a more general notion of monitorability
of an LTL formula: Intuitively, we say that the language given by an LTL
formula is \emph{monitorable} if, after any number of observed actions,
the monitor is still able to detect the violation or satisfaction of the
monitored property, and after at most finitely many additional
observations.  
%
%
As an example of a non-monitorable, LTL-definable language consider a
property such as ``it is always the case that a request will eventually
be answered,'' which is a so called \emph{liveness} property.  For this
property no finite witnesses of violation or satisfaction exist, since
any finite sequence of actions can be extended to satisfy this property.
In order to know that some request is, indeed, never answered, a monitor
would therefore require an infinite sequence of actions.  In
consequence, most examples of liveness properties that can be found in
the literature violate the intuitive definition of monitorability given
above.
To determine whether or not an LTL formula specifies a liveness property
is a PSpace-complete problem \cite{DBLP:journals/corr/cs-LO-0101017}.
However, they are not the only types of properties, which can be
formally specified in LTL that are not monitorable, and as this paper
will show there exists no criterion that allows to answer the
monitorability question for any given formula in a simple, syntactic
manner.


Pnueli and Zaks \cite{DBLP:conf/fm/PnueliZ06} were the first to
\emph{formalise} a notion of monitorability, which matches the intuitive
account given above:
According to \cite{DBLP:conf/fm/PnueliZ06} a formula is monitorable
wrt.\ a finite sequence of actions, if that finite sequence can be
extended to be a finite witness for violation or satisfaction of that
formula.  However, Pnueli and Zaks did not address the question of
deciding monitorability for a given formula (and sequence).
In \cite{tum-i0724} a slightly more general formalisation based on a
3-valued semantics for LTL is given, such that monitorability of an LTL
formula becomes a property of only the formula.
Moreover, Falcone et al.~\cite{verimag-TR-2009-6} have recently shown
that the definition given in \cite{tum-i0724} is, indeed, a
generalisation of the one given earlier in
\cite{DBLP:conf/fm/PnueliZ06}, and termed it ``classical
monitorability.''
In their paper, they have at first wrongly concluded---but later also
corrected \cite{DBLP:conf/rv/FalconeFM09}---that the class of
monitorable languages, under classical monitorability, consists exactly
of the \emph{obligation} properties in the hierarchy of safety-progress
properties (cf.\ \cite{93442}), which is orthogonal to the
safety-liveness classification.
An obligation property, for example, is obtained by taking a positive
Boolean combination of safety and co-safety properties. 
%
%
Despite their correction, Falcone et al.\ left the question regarding
the complexity of monitorability of an LTL formula (or $\omega$-regular
language in general) open.
%
%
Note that \cite{tum-i0724} did imply a decision procedure based on the
construction and subsequent analysis of deterministic monitors for LTL
formulae, but the given procedure requires 2ExpSpace (see
Sec.~\ref{sec:mon}).

One of the main contributions of this paper is a proof that this upper
bound is not optimal, in that monitorability of an LTL formula can be
decided in PSpace.  In fact, it will show that the monitorability
problem of LTL, i.e., the decision problem that asks ``is a given LTL
formula monitorable?'' is PSpace-complete, and that this result even
transfers to $\omega$-regular languages in general---regardless as to
whether they are given by an LTL formula, a nondeterministic \buchi
automaton, or an $\omega$-regular expression.
As such it is also proof that no simple syntactic categorisation of
monitorability of an LTL formula (or $\omega$-regular language), which
could be checked in polynomial time, exists.  
On the other hand, the result implies
%
that checking monitorability is no more complex than checking safety or
co-safety, which have often served as the ``monitorable fragment'' in
the past (cf.\
\cite{DBLP:journals/tissec/Schneider00,havelund02synthesizing,havelund2004}).

As a special case the paper also considers the monitorability problem of
\buchi automata, where the automaton in question is deterministic, and
shows that this restricted form of the problem is solvable in polynomial
time. 
%
Finally, it shows that the monitorable $\omega$-languages are closed
under the usual Boolean connectives; that is, they are closed under
finitary application of union, intersection, and complementation.

\paragraph{Outline.}
The remainder 
is structured as follows.
The next section recalls some preliminary notions and notations used
throughout this paper.
Sec.~\ref{sec:mon} gives a formal account of monitorability of an
$\omega$-language and phrases the corresponding decision problem(s).
Sec.~\ref{sec:class} puts two well-known classifications of
$\omega$-regular languages, namely the classification in terms of the
safety-progress hierarchy (cf.\ \cite{93442}) as well as a topological
view, in relation with the notion of monitorability.
The main contribution of this paper, which makes use of these
classifications, can be found in sections~\ref{sec:ltl} and
\ref{sec:nba} and as such they are also the most technical sections, in
that they contain the complexity analyses and proofs of the
monitorability problems of $\omega$-regular languages.
Sec.~\ref{sec:closure} details on closure properties of monitorable
$\omega$-langues, and Sec.~\ref{sec:end} concludes.

\section{Basic notions and notation}

\label{sec:basics}

We encode information about a system's state in terms of a finite set of
atomic propositions, $AP$, and define an \emph{action} to be an element
of $2^{AP}$.  In a sense, an action can be seen as a global state that
is determined by the individual atomic sub-states encoded by elements
from $AP$.  We will therefore use the terms action and state
synonymously.  The system behaviour which the monitor observes then
consists of a sequence of actions.  Therefore, we define an
\emph{alphabet}, $\Sigma := 2^{AP}$, and treat consecutive sequences of
actions as \emph{words} over $\Sigma$.
As is common, we define $\Sigma^\ast$ as the set of all finite words
over $\Sigma$, including the empty word, and $\Sigma^\omega$ to be the
set of infinite words obtained by concatenating an infinite sequence of
nonempty words over $\Sigma$.
%
Infinite words are of the form $w = w_0w_1\ldots \in \Sigma^\omega$ and
are usually abbreviated by $w, w'$, and so on, whereas finite words are
of the form $u = u_0\ldots u_n \in \Sigma^\ast$ and are usually
abbreviated by $u, u', v$, and so on.  Let $w \in \Sigma^\omega$, then
$w^i$ denotes the infinite suffix $w_iw_{i+1}\ldots$, whereas $u \preceq
w$ denotes a \emph{prefix} of $w$.
$u$ is a \emph{proper prefix} of $w$ ($u \prec w$), if $u \preceq w$ and
$u \neq w$.
For any $p \in AP$, and a given $\sigma \in \Sigma$, if $p \in \sigma$
holds, we also say that ``$p$ holds (or, is true) in the state
$\sigma$''.  If $p \not\in \sigma$, then ``$p$ does not hold (or, is not
true) in state $\sigma$.''

The syntax of LTL formulae, which are given by the set $\LTL(AP)$, is
defined as follows: $\varphi ::= p \mid \neg \varphi \mid \varphi \vee
\varphi \mid \ltlX \varphi \mid \varphi \ltlU \varphi$, with $p \in AP$.
If the set of atomic propositions is clear from the context, we  
write LTL instead of $\LTL(AP)$.  LTL formulae are interpreted over
elements from $\Sigma^\omega$ as follows.  Let $i \in N$, and
$\varphi,\psi \in \LTL$, then
\[
\begin{array}{lcl}
  w^i \models p & \Leftrightarrow & p \in w_i\\
  w^i \models \neg\varphi & \Leftrightarrow & w^i \not\models \varphi\\
  w^i \models \varphi \vee \psi &\Leftrightarrow & w^i \models \varphi \vee w^i \models \psi\\
  w^i \models \ltlX\varphi & \Leftrightarrow & w^{i+1} \models \varphi\\
  w^i \models \varphi \ltlU \psi &\Leftrightarrow& \exists k \geq i.\ w^k \models \psi \wedge \forall i \leq j < k.\ w^j \models \varphi
\end{array}
\]
Further, we will make use of the usual syntactic sugar such as $true
\equiv p \vee \neg p$, $false = \neg true$, $\varphi \wedge \psi \equiv
\neg(\neg \varphi \vee \neg \psi)$, $\ltlF\varphi \equiv true \ltlU
\varphi$, and $\ltlG\varphi \equiv \neg (\ltlF\neg \varphi)$.
%

It is well-known that, for any $\varphi \in \LTL$, we can construct a
\emph{nondeterministic \buchi automaton} (NBA), $\cAphi = (\Sigma, Q,
Q_0, \delta, F)$, where $\Sigma$ is the alphabet, $Q$ the set of states,
$Q_0 \subseteq Q$ designated initial states, $\delta: Q \times \Sigma
\rightarrow 2^Q$ the transition relation, and $F \subseteq Q$ a set of
final states, such that the accepted language of $\cAphi$ contains
exactly all the models of $\varphi$, i.e., $\cL(\cAphi) = \cL(\varphi)$.
If some language of infinite words, called an $\omega$-language,
$L \subseteq \Sigma^\omega$ is such that there exists an NBA, $\cA$,
such that $\cL(\cA) = L$, then $L$ is called
\emph{$\omega$-regular}. Obviously, the language specified by an LTL
formula is always $\omega$-regular.
The size of $\cAphi$, usually measured wrt.\ $|Q|$, is, in the
worst-case, exponential wrt.\ the size of $\varphi$.
For details on the construction as well as further properties of
$\cAphi$, cf.\ \cite{114895}.


\section{When is an $\omega$-language monitorable?}
\label{sec:mon}

Let us fix an 
$L \subseteq \Sigma^\omega$ for the remainder of
this section.  In accordance with \cite{DBLP:conf/fm/PnueliZ06} and
\cite{tum-i0724}, Falcone et al.~\cite{verimag-TR-2009-6} formally
define the monitorability of an $\omega$-language as follows.

\begin{definition}
  \label{def:mon1}
  $L$ is called
  \begin{itemize}
  \item \emph{negatively determined} by $u \in \Sigma^\ast$, if 
    $u\Sigma^\omega \cap L = \emptyset$;
  \item \emph{positively determined} by $u \in \Sigma^\ast$, if
    $u\Sigma^\omega \subseteq L$;
  \item \emph{$u$-monitorable} for $u \in \Sigma^\ast$, if $\exists v
    \in \Sigma^\ast$, s.t.\ $L$ is positively or negatively determined
    by $uv$;
  \item \emph{monitorable}, if it is $u$-monitorable for any $u \in
    \Sigma^\ast$.
  \end{itemize}
\end{definition}
This also lends itself to another, sometimes more intuitive way to think
about monitorability of an $\omega$-language, namely in terms of
\emph{good} and \emph{bad} prefixes.

\begin{definition}
  The set of \emph{good} and \emph{bad prefixes} for $L$ are defined as
  $
  good(L) := \{ u \in \Sigma^\ast \mid u\Sigma^\omega \subseteq L \}
  $
  and
  $bad(L) := \{ u \in \Sigma^\ast \mid u \Sigma^\omega \cap L = \emptyset \}$,
  respectively.
\end{definition}
For brevity, we  also write $good(\varphi)$ (respectively,
$bad(\varphi)$) short for $good(\cL(\varphi))$ (respectively,
$bad(\cL(\varphi))$), and $good(\cA)$ (respectively, $bad(\cA)$) short
for $good(\cL(\cA))$ (respectively, $bad(\cL(\cA))$).

\begin{proposition}
  \label{prop:mon}
  $L$ 
  is \emph{monitorable} if
  $\forall u \in \Sigma^*.\ \exists v \in \Sigma^*.\ uv \in good(L) \vee
  uv \in bad(L)$.
  %
\end{proposition}
In other words, $L$ is \emph{not} monitorable if there exists a finite
word $u\in\Sigma^\ast$ for which we can not find a finite extension $v
\in \Sigma^\ast$, such that $uv$ is either a good or a bad prefix of
$L$.  
Naturally, given some $L$, not every finite word is a good or a bad
prefix of $L$, in which case we call such a word \emph{undetermined}
(wrt.\ $L$).  Let $u \in \Sigma^\ast$ be an undetermined prefix, then,
depending on $L$, the following scenarios are possible: we can find a
finite extension $v \in \Sigma^\ast$, such that $uv \in good(L)$, we can
find a finite extension $v$, such that $uv \in bad(L)$, or there does
not exist a finite extension $v$, such that $uv \in good(L)$ or $uv \in
bad(L)$ would hold.
In \cite{tum-i0724}, the latter were called ``ugly'' prefixes, and $L$
``non-monitorable,'' if there exists an ugly prefix for it.

Let us now 
define the monitorability problem of an
$\omega$-language as follows.  
%
%
\begin{definition}
  The \emph{monitorability problem} for some
  $L$ 
  is the
  following decision problem:\\
  \textbf{Given:} A set $L \subseteq \Sigma^\omega$.\\
  \textbf{Question:} Does $\forall u \in \Sigma^*.\ \exists v \in
  \Sigma^*.\ uv \in good(L) \vee uv \in bad(L)$ hold?
\end{definition}
When $L$ is given in terms of an LTL formula, an NBA, or an
$\omega$-regular expression (which are basically defined like ordinary
regular expressions, augmented with an operator for infinite repetition
of a regular set, cf.\ \cite{114895}), we call this problem the
\emph{monitorability problem of $\omega$-regular languages}, or---more
specifically---the \emph{monitorability problem of LTL/\buchi
  automata/$\omega$-regular expressions}, respectively.
%


One of the main contributions of \cite{tum-i0724} was a procedure that,
given a formula $\varphi \in \LTL$, constructs a deterministic
finite-state machine (i.e., a monitor for $\varphi$) whose input is a
consecutively growing, finite word $u \in \Sigma^\ast$, and whose output
is $\top$ if $u \in good(\varphi)$, $\bot$ if $u \in bad(\varphi)$, and
$?$ if $u$ is undetermined.
Once this monitor is computed, the monitorability of $\varphi$ can be
determined in polynomial time, simply by checking if there exists a
state whose output is $?$ with no path leading to a $\top$- or
$\bot$-state.  If such a state, called a $?$-trap, exists, then
$\varphi$ is not monitorable.
Notice, however, that this monitor construction (and this decision
procedure) requires 2ExpSpace:
as a first step, it creates two NBAs, one which accepts all models of
$\varphi$ and one that accepts all counterexamples of $\varphi$ (i.e.,
all models of $\neg\varphi$), and then proceeds by examining and
transforming the resulting state graphs of these automata.  Recall, NBAs
accepting the models of an LTL formula are, in the worst case,
exponentially larger than the corresponding formula.  Since at some
point, the two automata are made deterministic, the double exponential
``blow up'' follows.
Moreover, although not explicitly mentioned in \cite{tum-i0724}, by
altering the first step of this procedure, it can be used to decide the
monitorability of $\omega$-regular languages, in general, i.e., whether
given as an LTL formula, as NBA, or as an $\omega$-regular expression.
For example, if instead of a formula, an NBA is given, one has to
explicitly complement this automaton, which also involves a worst-case
exponential ``blow up'' wrt.\ the number of states of the original NBA.
However, then the rest of the procedure described in \cite{tum-i0724}
stays the same.  On the other hand, if we are given an $\omega$-regular
expression instead, we first have to build an NBA, which can occur in
polynomial time.  Then, in order to get the complementary language, one
also needs to complement this automaton.  Hence, independent of the
concrete representation of an $\omega$-regular language, the
construction and subsequent analysis of the corresponding monitor can
decide monitorability in 2ExpSpace.
Therefore, indirectly, \cite{tum-i0724} shows decidability of the
monitorabiliy problem, but whether or not this bound is tight was left
open in that paper. 

\paragraph{Examples.}
Let us examine some examples to understand how this construction works
and what its outcome is.
Fig.~\ref{fig:ex} depicts some finite state machines (i.e., the
monitors) for several LTL formulae, which were automatically generated
using the \LTLtools\footnote{Available under an open source license at
  \textsf{http://LTL3tools.SourceForge.Net/}}, which are written by the
author of this paper and implement the above construction.
Each monitor is complete in a sense that for every action from the
alphabet, there exists a transition.  Note that, although not explicitly
marked, the initial state is the top-most $?$-state, respectively.  Any
word $u\in\Sigma^\ast$ which has a corresponding path in a monitor to a
$?$-state is undetermined wrt.\ the $\omega$-language being monitored.
On the other hand, if $u$ leads to a state labelled $\top$
(respectively, $\bot$), then $u$ is a good (respectively, bad) prefix
of the $\omega$-language being monitored.
\begin{figure}[tb]
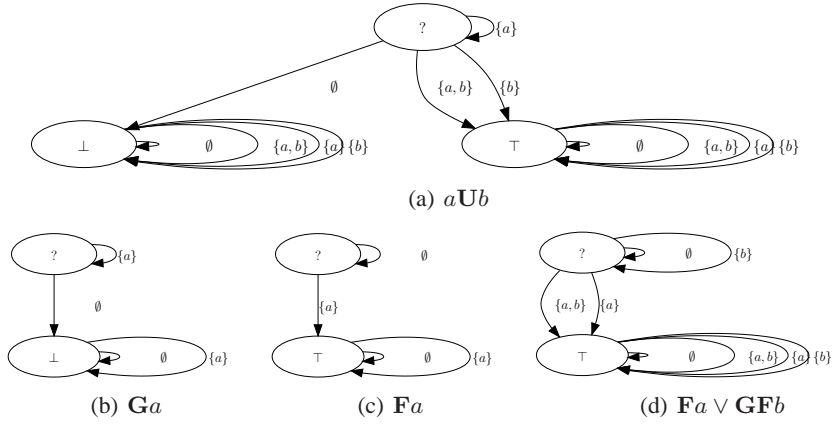

  \centering
  \subfigure[$a \ltlU b$]{\label{fig:ex:aUb}
    \scalebox{0.4}{\input{aUb.tex}}}
  \subfigure[$\ltlG a$]{\label{fig:ex:Ga}
    \scalebox{0.35}{\input{Ga.tex}}}
  \subfigure[$\ltlF a$]{\label{fig:ex:Fa}
    \scalebox{0.35}{\input{Fa.tex}}}
  \subfigure[$\ltlF a \vee \ltlG\ltlF b$]{\label{fig:ex:FaGFb}
    \scalebox{0.35}{\input{FaOrGFb.tex}}}
  \caption{Some example monitors and corresponding LTL specifications.}
  \label{fig:ex}
\end{figure}
It is easy to see, that all the formulae give rise to a monitorable
language; that is, from any reachable state in the respective monitor,
there always exists a path to a state labelled either $\top$ or $\bot$.
Let us, therefore, also present a language which is not monitorable and
whose (practically not very useful) monitor is depicted in
\figref{fig:ex:non}:
\begin{figure}[tb]
  \centering
  \scalebox{0.4}{\input{aAndXGFb.tex}}
  \caption{``Monitor'' for a non-monitorable language given by $a \wedge
    \ltlX (\ltlG\ltlF b)$.}
  \label{fig:ex:non}
\end{figure}
Clearly, the right-most state is a $?$-trap; that is, once reached by
some finite prefix $u \in \Sigma^\ast$, there exists no extension $v \in
\Sigma^\ast$ for $u$, such that a $\top$- or a $\bot$-state can be
reached.
Or, in other words, every word $u = u_0 \ldots u_n$, such that $a \in
u_0$ is an ugly prefix of $\cL(a \wedge \ltlX (\ltlG\ltlF b))$.

%


\section{A classification of $\omega$-languages}
\label{sec:class}

\subsection{The safety-liveness view}

\label{sec:props}

Alpern and Schneider \cite{alpern87recognizing} were the first to give a
formal characterisation of $\omega$-languages in terms of \emph{safety}
and \emph{liveness} properties.  This view was subsequently extended to
an entire hierarchy of $\omega$-languages, where languages defining
safety properties, or their complement are at the bottom (cf.\
\cite{93442}).

\begin{definition}
  $L$ describes a \emph{safety}
  language (also called a safety property), if $\forall w \not\in L.\
  \exists u \prec w.\ u\Sigma^\omega \cap L = \emptyset$.
  $L$ describes a \emph{co-safety} language (also called co-safety
  property), if $\forall w \in L.\ \exists u \prec w.\ u\Sigma^\omega
  \subseteq L$.
\end{definition}
In other words, if $L$ specifies a safety language, then all infinite
words $w \not\in L$, have a bad prefix.  On the other hand, if $L$
specifies a co-safety language, then all infinite words $w \in L$, have
a good prefix.
This also explains why safety and co-safety properties lend themselves
so well to runtime verification using monitors: if the specification to
be monitored gives rise to a safety language, then all violations of the
specification are detectable by the monitor after only finitely many
observations of actions emitted by the system under scrutiny.
That is, let $u \in \Sigma^\ast$ be a word, resembling the sequence of
actions, then either there exists a $v\in\Sigma^\ast$, such that $uv \in
bad(L)$, or $u \in good(L)$ already holds.  
On the other hand, if the specification to be monitored gives rise to a
co-safety language, then all models of the specification are detectable
by the monitor after only finitely many observations\footnote{From this
  point forward, we will omit the use of precisifications such as
  ``emitted by the system under scrutiny'', etc.\ and simply speak of
  abstract words and actions, when the context is clear or simply does
  not matter}.
That is, let $u \in \Sigma^\ast$, then either there exists a $v \in
\Sigma^\ast$, such that $uv \in good(L)$, or $u \in bad(L)$ already
holds.
%
%
It follows that a co-safety language always has the form
$B\Sigma^\omega$, where $B\subseteq \Sigma^\ast$.

Safety and co-safety are dual in a sense that if $L$ is a safety
language, then $\Sigma^\omega \backslash L$, from this point forward
also abbreviated as $\overline{L}$, is a co-safety language, and vice
versa.
The following easy to prove proposition makes this duality formal.

\begin{proposition}
  \label{prop:prefixes}
  $bad(L) = good(\overline{L})$ and $bad(\overline{L}) = good(L)$.
\end{proposition}

\begin{definition}
  $L$ describes a \emph{liveness} language (also called a liveness
  property) if $\forall u \in \Sigma^\ast.\ u\Sigma^\omega \cap L \neq
  \emptyset.$
\end{definition}
In other words, if $L$ specifies a liveness language, then $bad(L) =
\emptyset$---in which case $L$ may only be monitorable if $good(L) \neq
\emptyset$ also holds.  The definition of liveness, however, does not
require $good(L)$ to be empty or non-empty.  Hence, from a runtime
verification point of view, many liveness languages which are commonly
used to describe system properties in the area of formal verification
using, say, temporal logic model checking, turn out to be not
monitorable.

\paragraph{Examples.}
Let us look at some example languages, specified in terms of LTL
formulae.  The formula $\varphi = \ltlG \neg\mathsf{bad\_state}$ with
$\mathsf{bad\_state} \in AP$ formalises, in an abstract manner, the
requirement from the introduction: the system must never enter a bad
state.  In other words, $\neg\mathsf{bad\_state}$ must always be true.
It is a safety property as any prefix containing a state in which
$\mathsf{bad\_state}$ is true is a bad prefix of $\varphi$, e.g.,
$
u = \emptyset\emptyset\emptyset\emptyset\ldots\{\mathsf{bad\_state}\} \in bad(\varphi).
$
Naturally, $\neg\varphi = \neg (\ltlG \neg\mathsf{bad\_state}) = \ltlF
\mathsf{bad\_state}$ describes a co-safety property.  Any finite word
containing a state where the proposition $\mathsf{bad\_state}$ is true
is a good prefix for $\neg\varphi$. 
In practical terms, this means that a monitor, checking either language
will be able to make a conclusive decision after the first occurrence of
$\mathsf{bad\_state}$ in the observed sequence of system actions.
In fact, we can postulate the following proposition which, using
Proposition~\ref{prop:prefixes}, is easy to prove formally:
\begin{proposition}
  \label{prop:safe}
  If $\varphi$ specifies a safety or a co-safety language, then
  $\varphi$ is monitorable.
\end{proposition}
It is also easy to verify that the formula $\ltlF\mathsf{bad\_state}$
meets the definitions of both co-safety and liveness.
However, as we will see, not all co-safety languages are also liveness
languages, and vice versa.  In fact, unlike $\ltlF\mathsf{bad\_state}$
most liveness languages do not lend themselves to runtime verification
via monitors, because they may have neither bad nor good prefixes that
would eventually lead a monitor to a conclusive answer.
The liveness property given in the introduction, and formalised in LTL
as $\ltlG(\mathsf{request} \rightarrow \ltlF \mathsf{answer})$, is such
a case.
The following proposition is 
 easy to prove:
\begin{proposition}
  \label{prop:live}
  If $\varphi$ specifies a liveness language, such that $good(\varphi) =
  \emptyset$, then $\varphi$ is not monitorable.
\end{proposition}
It follows that a formula of the form $\varphi = \ltlG\ltlF (\psi)$ is
not monitorable unless $\psi = true$ or $\psi = false$.  In the first
case we would get $\cL(\varphi) = \Sigma^\omega$, and in the latter
$\cL(\varphi) = \emptyset$, both of which meet the definition of
monitorability.  In fact, the languages given by the sets $\emptyset$
and $\Sigma^\omega$ are both safety and co-safety.

As a final example, let us consider \emph{obligation} languages.  In
\cite{93442}, Manna and Pnueli define the class of obligation languages
as follows.
\begin{definition}
  $L$ describes an \emph{obligation} language (also called an obligation
  property) if $L$ either consists of an unrestricted Boolean
  combination of safety languages, or an unrestricted Boolean
  combination of co-safety languages, or a positive Boolean combination
  of safety and co-safety languages.
\end{definition}
Falcone et al.~\cite{DBLP:conf/rv/FalconeFM09} have shown that
\begin{proposition}
  \label{prop:obligation}
  If $\varphi$ specifies an obligation language, then $\varphi$ is
  monitorable.
\end{proposition}
To see that the other direction 
is not true,
consider the counterexample given by the formula in
\figref{fig:ex:FaGFb}: it does not specify an obligation language, yet
it is monitorable.



\subsection{The corresponding topological view}

\label{sec:top}

Alpern and Schneider 
showed in \cite{alpern87recognizing} that
%
\begin{proposition}
  \label{prop:dt}
  Every language $L$ can be represented as the intersection
  $L = L_S \cap L_L,$
  where $L_S$ is a safety language, and $L_L$ is a liveness language.
\end{proposition}
Their proof is based on the observation that safety languages (over some
alphabet $\Sigma$) correspond to \emph{closed} sets in the Cantor
topology over $\Sigma^\omega$ (cf.\ \cite{93442}), and liveness
languages to dense sets.  It follows that co-safety languages correspond
to \emph{open} sets in that topology.  Sets which are both closed and
open, are referred to as clopen.  It is worth pointing out, and easy to
prove, that both $\emptyset$ and $\Sigma^\omega$ are \emph{clopen}.
%
Given a set $L\subseteq\Sigma^\omega$ and element $w \in \Sigma^\omega$,
$w$ is a \emph{limit point} of $L$, if there exists an infinite sequence
of words $w_1, w_2, \ldots$, all of which are in $L$, which converges to
$w$.  Clearly, any $w \in L$ is a limit point of $L$, since $w, w,
\ldots$ converges to $w$.  The \emph{topological closure} of $L$,
written $cl(L)$, is then defined as the set of all limit points of $L$.
Then, obviously, $L \subseteq cl(L)$.
%
The following gives a direct definition of $cl(L)$:
\begin{definition}
  $cl(L) := \{ w \in \Sigma^\omega \mid \forall u \prec w.\ \exists w'
  \in L.\ u \prec w' \}$.
\end{definition}
From a basic result of topology, a topological closure operator on
$\Sigma^\omega$ defines a topology, where a set $L \subseteq
\Sigma^\omega$ is closed (i.e., a safety language) if and only if 
%
$cl(L) \subseteq L$ also holds.
%
%
Moreover, $L$ is \emph{dense} (i.e., a liveness language), if and only
if $cl(L) = \Sigma^\omega$.

This alternative classification of $\omega$-languages proved useful as
many important results from topology transfer to the commonly used
classification in terms of safety and liveness properties.
For example, due to Alpern and Schneider \cite{alpern87recognizing} it
is well-known that the topological closure of a language that is given
by an NBA, $\cA$, where non-reachable and dead-end states have been
eliminated, can be determined by an NBA, $\cA'$, which is like $\cA$
except that all states are made final.  Now, $\cL(\cA)$ gives rise to a
safety language if and only if $\cL(\cA) = \cL(\cA')$ as $\cL(\cA') =
cl(\cL(\cA))$.
We will make use of this and similar results in the remainder. 
For a 
comprehensive overview on this topology, cf.\ 
\cite{alpern87recognizing,93442}.


\section{The monitorability problem of LTL}
\label{sec:ltl}

The results of this section will show that the monitorability problem of
LTL is PSpace-complete.
In order to show this, we will make use of a well-known construction of
a tableau for an LTL formula, which has been given many times before in
the literature (cf.\
\cite{DBLP:journals/jacm/SistlaC85,DBLP:journals/fac/Sistla94}).
For reasons of self-containedness, we briefly summarise its most
important properties for our purposes.

Let us first fix a formula $\varphi \in \LTL$ over some alphabet
$\Sigma$.
$SF(\varphi)$ is the set consisting of the subformulae of $\varphi$ or
the negations of subformulae of $\varphi$.  A set $c \subseteq
SF(\varphi)$ is \emph{complete} if the following two conditions are met:
1.\ Boolean consistency of $c$;
%
2.\ for $\varphi' = \mu \wedge \nu \in SF(\varphi)$, $\varphi' \in c$ if
and only if $\mu \in c$ and $\nu \in c$.
Let $tab(\varphi) = (V, E)$ be a directed graph, where $V$ is the set of
all complete subsets of $SF(\varphi)$, and elements $(c, d) \in E$
defined as follows:
\begin{itemize}
\item for any $\varphi' = \mu \ltlU \nu \in SF(\varphi)$: $\varphi' \in
  c$ if and only if 
  $\nu\in c$, or $\mu\in c$ and $\varphi' \in d$;
\item for any $\varphi' = \ltlX \psi \in SF(\varphi)$: $\varphi' \in
  c$ if and only if $\psi \in d$.
\end{itemize}
Let for any $c\in V$, $\pi(c)$ be the state such that for any atomic
proposition $p \in SF(\varphi)$, $\pi(c)(p) = true$ if and only if $p
\in c$.  An infinite path through $tab(\varphi)$ is called
\emph{accepting} if for every node $c$ on that path with $\varphi' = \mu
\ltlU \nu \in c$, either $\nu \in c$, or there exists a (not necessarily
immediate) successor node $d$, such that $\nu \in d$.
For any $c\in V$, we say that $c$
is a \emph{good node}, if the conjunction of all subformulae in $c$ is
satisfiable; otherwise $c$ is called a \emph{bad node}.
Notably, it holds that for any $w\in\Sigma^\omega$ with the property
$\forall i \geq 0.\ \exists w' \in \Sigma^\omega$ such that $w_0\ldots
w_i w' \models \varphi$, there exists an infinite path $\rho$ of good
nodes in $tab(\varphi)$ starting from a node that contains $\varphi$,
such that $\pi(\rho) = w$.
Moreover for $\varphi, \psi \in \LTL$, we denote by $tab(\varphi) \times
tab(\psi)$ the cross-product of the tableaux for $\varphi$ and $\psi$,
respectively.

\begin{lemma}
  \label{lem:ltl:1}
  Let $\varphi$ not be monitorable.  Then there exists a pair of nodes,
  $(q, q') \in tab(\varphi) \times tab(\neg\varphi)$, reachable on some
  $u \in \Sigma^\ast$ and where $q,q'$ are conjunctions of subformulae
  of $\varphi$, respectively, such that $\cL(q)$ and $\cL(q')$ are
  dense.
\end{lemma}

\begin{proof}
  Following Proposition~\ref{prop:mon}, the \emph{non-monitorability} of
  $\varphi$ is defined as follows
  $$
  \exists u \in \Sigma^\ast.\ \forall v \in \Sigma^\ast.\
  uv\Sigma^\omega \cap \cL(\varphi) \neq \emptyset \wedge
  uv\Sigma^\omega \not\subseteq \cL(\varphi).
  $$
  In other words, there exists a $u \in \Sigma^\ast$, such that none of
  the finite continuations $v$ of $u$ is $(i)$ a bad or $(ii)$ a good
  prefix of $\varphi$.
  Let us fix such a particular $u$.
  From the construction of $tab(\varphi)$ it follows that in order for
  $(i)$ to be true, there must exist a node $q \in V_\varphi$, reachable
  on $u$ (i.e., $tab(\varphi)$ has a path on $u$), such that $ \forall v
  \in \Sigma^\ast.\ v\Sigma^\omega \cap \cL(q) \neq \emptyset.$ It is
  easy to see that $\cL(q)$ is dense.
  Requirement $(ii)$, i.e., $\forall v \in \Sigma^\ast.\ uv\Sigma^\omega
  \not\subseteq \cL(\varphi)$, is equivalent to 
  $\forall v \in \Sigma^\ast.\ uv \not\in good(\varphi)$, which by
  Proposition~\ref{prop:prefixes} is equivalent to
  \begin{equation}
    \forall v \in \Sigma^\ast.\ uv \not\in bad(\neg\varphi).
    \label{eq:ltl1}
  \end{equation}
  Let $q' \in V_{\neg\varphi}$ be a node in $tab(\neg\varphi)$,
  reached on $u$.
  %
  %
  Now, for (1) to be true, $\forall v \in \Sigma^\ast.\ v \not\in
  bad(q')$ must be true, which is equivalent to $\forall v \in
  \Sigma^\ast.\ v\Sigma^\omega \cap \cL(q') \neq \emptyset$.
  It is easy to see that $\cL(q')$ is dense.
  %
  \qed
\end{proof}

\begin{lemma}
  \label{lem:ltl:2}
  If there exists a pair $(q, q') \in tab(\varphi) \times
  tab(\neg\varphi)$, reachable on some $u \in \Sigma^\ast$, such that
  $\cL(q)$ and $\cL(q')$ are dense, then $\varphi$ is not monitorable.
\end{lemma}

\begin{proof}
  Let $(q, q') \in tab(\varphi) \times tab(\neg\varphi)$ be reached via
  some $u\in \Sigma^\ast$, such that
  \[
  \forall v \in \Sigma^\ast.\ v \not\in bad(\cL(q)) \wedge v \not\in bad(\cL(q')).
  \]
  Since $q$ is reached on $u$, and by the construction of $tab(\varphi)$
  it follows that $\forall v \in \Sigma^\omega.\ v\not\in bad(\cL(q))$
  is equivalent to $\forall v \in \Sigma^\omega.\ uv \not\in
  bad(\varphi)$ (and, accordingly, for $q'$ and $\neg\varphi$).  Thus,
  together with Proposition~\ref{prop:prefixes}
  we get
  $
  \exists u \in \Sigma^\ast.\ \forall v \in \Sigma^\ast.\ uv \not\in
  bad(\varphi) \wedge uv\not\in good(\varphi),
  $
  which corresponds to the definition of non-monitorability of
  $\varphi$, used in the previous lemma.
  \qed
\end{proof}

\begin{theorem}
  The monitorability problem of LTL is decidable in PSpace.
\end{theorem}

\begin{proof}
  By Lemma~\ref{lem:ltl:1} and~\ref{lem:ltl:2}, $\varphi$ is not
  monitorable if and only if there exists a word, corresponding to a
  path through $tab(\varphi) \times tab(\neg\varphi)$ that contains a
  pair $(q, q')$, such that $\cL(q)$ and $\cL(q')$ are dense.
  As $tab(\varphi)$ and $tab(\neg\varphi)$ are of exponential size wrt.\
  $|\varphi|$, we cannot construct either explicitly.  Instead, we will
  guess, in a step-wise manner, a path through $tab(\varphi) \times
  tab(\neg\varphi)$ to some pair $(q, q')$, and check if both $\cL(q)$
  and $\cL(q')$ are dense.
  To check whether or not an \LTL formula specifies a dense set is
  equivalent to checking whether or not it specifies a liveness
  language (cf.\ \secref{sec:top}).
  It follows from Ultes-Nitsche and Wolper's work
  \cite{DBLP:journals/corr/cs-LO-0101017} (Remark 4.3 and Theorem 4.6,
  if we replace $L_\omega$ to correspond to $\Sigma^\omega$) that this
  problem can be decided in PSpace.
  So, if the answer to this check is ``yes'', then $\varphi$ is not
  monitorable.

  Since due to Savitch's theorem we know that NPSpace is equal to PSpace
  (cf.\ \cite{cj97-01-15}), we have thus shown that the
  ``non-monitor\-ability problem of LTL'' is in PSpace.
  However, as PSpace is equal to co-PSpace \cite{cj97-01-15}, it follows
  that the complementary problem of that, i.e., the monitorability
  problem of LTL, is, in fact, decidable in PSpace.
  \qed
\end{proof}

\begin{theorem}
  The monitorability problem of LTL is PSpace-complete.
\end{theorem}

\begin{proof}
  It is sufficient to show PSpace-hardness.
  We will reduce the PSpace-complete problem of determining whether or
  not a formula $\varphi \in \LTL$ is satisfiable
  \cite{DBLP:journals/jacm/SistlaC85} to the monitorability problem of
  LTL. 
  Let us construct, in constant time, a formula $\psi \in LTL(AP') :=
  \ltlG a \vee \ltlG\ltlF (a' \wedge \varphi)$, where $a\in AP$, $AP' :=
  AP \cup \{a'\}$ and $a' \not\in AP$.  We now claim that $\psi$ is
  monitorable if and only if $\cL(\varphi) = \emptyset$.

  If $\cL(\varphi) = \emptyset$, then $\psi \equiv \ltlG a$, which can
  easily be seen monitorable.

  For the other direction, assume that $\psi$ is monitorable, but that
  $\cL(\varphi) \neq \emptyset$.
  Let $\Sigma' := 2^{AP'}$ and $u \in (2^{AP\backslash\{a\}})^\ast$.  It
  is easy to see that $u\Sigma'^\omega \cap \cL(\ltlG a) = \emptyset$,
  but $u\Sigma'^\omega \cap \cL(\psi) \neq \emptyset$.  Hence, for
  $\psi$ to be monitorable, $u$ has to be extensible with some
  $v\in\Sigma'^\ast$, such that either $uv\Sigma'^\omega \cap
  \cL(\ltlG\ltlF (a' \wedge \varphi)) = \emptyset$, or such that
  $uv\Sigma'^\omega \subseteq \cL(\ltlG\ltlF (a' \wedge \varphi))$.
  Now, observe that irrespective of our choice of $\varphi$ (including
  the case $\varphi = true$), so long as $\cL(\varphi) \neq \emptyset$,
  the set $\cL(\ltlG\ltlF (a' \wedge \varphi))$ neither has a bad nor a
  good prefix.
  %
  %
  This means that $\exists u \in \Sigma'^\ast.\ \forall v \in
  \Sigma'^\ast.\ uv\Sigma'^\omega \cap \cL(\psi) \neq \emptyset \wedge
  uv\Sigma'^\omega \not\subseteq \cL(\psi)$; that is, $\psi$ is not
  monitorable.  Contradiction.  \qed

\end{proof}


\section{The monitorability problem of \buchi automata}
\label{sec:nba}


Let for the rest of this section $\cA = (\Sigma, Q, Q_0, \delta, F)$ be
a fixed NBA with $\cL(\cA) \subseteq \Sigma^\omega$.  As pointed out in
\secref{sec:top}, a topologically closed set can be obtained from
$\cL(\cA)$ via an automaton, referred to as $safe(\cA)$, whose accepted
$\omega$-language will always be closed, irrespective of $\cL(\cA)$.
What is more, $safe(\cA)$ can be constructed in polynomial time wrt.\
the size of $\cA$.
Moreover for the next results, we also need an explicit representation
of $live(\cA)$, whose accepted $\omega$-language will always be dense.
Like its counterpart $safe(\cA)$, it can be constructed in polynomial
time wrt.\ the size of $\cA$.  For details on these constructions, see
\cite[Sec.~4]{alpern87recognizing}.
Finally, we need to introduce the notion of a \emph{tight automaton}, as
a finite-state acceptor for good prefixes, as follows.


\subsection{Tight automata}

In preparation for the main results of this section, let us discuss how
to obtain a \emph{tight automaton} over $\Sigma^\ast$ that, given some
NBA $\cA$, accepts $good(\cA)$.
%
%
The construction can be described by a two-stage process.

First, we construct a nondeterministic finite automaton (NFA) that
accepts the \emph{potentially} good prefixes of $\cA$, where, without
loss of generality, all unreachable and dead-end states have been
eliminated.
From $\cA$, we can easily derive the NFA $\cG^p_\cA = (\Sigma, Q, Q_0,
\delta, F)$, where $F := Q$ is the set of accepting states, and the rest
defined as for $\cA$.  For this NFA it holds that
\begin{proposition}
  $\cL(\cG^p_\cA) = 
  \{ u \in \Sigma^\ast \mid \exists w \in \Sigma^\ast.\ uw \in \cL(\cA) \}$.
\end{proposition}
From this point forward, let as a notational convention, $\cA(q)$ be
like $\cA$, except that $Q_0 = \{ q \}$.
\begin{proof}
  ($\subseteq$):
  Take any $u \in \cL(\cG^p_\cA)$.  Obviously, there exists an accepting
  run on $u$ in $\cG^p_\cA$ to some state $q \in Q$ and by construction
  also a (finite) run in $\cA$ reaching the same $q$ as both automata
  share the same $\delta$.
  As by assumption $\cA$ is non-empty, and all unreachable and dead-end
  states have been eliminated, it then follows that $\cL(\cA(q)) \neq
  \emptyset$, i.e., $\exists w \in \Sigma^\omega.\ w \in \cL(\cA(q))$.
  Moreover, as $q$ was reached on $u$, it then follows that $uw \in
  \cL(\cA)$.

  ($\supseteq$):
  Let $w \in \Sigma^\omega$ be such that there exists an accepting run
  in $\cA$, i.e., $w \in \cL(\cA)$.  As $\delta$ is the same for both
  automata, it follows that there also exists a run on the state space
  of $\cG^p_\cA$, in a sense that for each symbol in $w$ there always
  exists a successor state in $\cG^p_\cA$.  Now, pick any $u \prec w$,
  then $u \in \cL(\cG^p_\cA)$ as all states in $\cG^p_\cA$ are
  accepting.
  \qed
\end{proof}
%

Second, to obtain an NFA that contains \emph{only} the good prefixes,
but no other words, we proceed as follows.
As $\cG^p_\cA$ is but an ordinary NFA, we can apply the standard subset
construction to obtain a deterministic finite automaton (DFA) accepting
the same language, and whose states consist of a subset of states of
$\cG^p_\cA$, respectively.
Let $\cG_\cA = (\Sigma, Q', \langle Q_0 \rangle, \delta', F')$ be this
DFA, defined as expected, except that we set the accepting states to be
\[
F' := \{ \langle q_0, \ldots, q_n \rangle \in Q' \mid \cL(\cA(q_0)) \cup
\ldots \cup \cL(\cA(q_n)) = \Sigma^\omega \}.
\]
Note that as a notational convention we let $\langle q_0, \ldots, q_n
\rangle$ be the single DFA-state whose label is made up of the
individual state labels $q_0, \ldots, q_n$ of $\cG^p_\cA$.

\begin{proposition}
  \label{prop:nba:tight}
  $\cL(\cG_\cA) = good(\cA)$.
\end{proposition}

\begin{proof}
  ($\subseteq$):
  Take any $u \in \cL(\cG_\cA)$.
  By the subset construction and the fact that all states in
  $\cL(\cG^p_\cA)$ are accepting, it follows that $u \in \cL(\cG^p_\cA)$
  must hold.  Hence, $u$ is a potentially good prefix of $\cL(\cA)$.
  Now, recall that for $u \in \cL(\cG_\cA)$ to hold, $\cG_\cA$ must be
  in some state $\langle q_0, \ldots, q_n \rangle$ such that
  $\cL(\cA(q_0)) \cup \ldots \cup \cL(\cA(q_n)) = \Sigma^\omega$ holds.
  Moreover, by the construction of $\cL(\cG^p_\cA)$, we know that there
  exist $n+1$ runs in $\cA$ on $u$ to the individual states $q_0,
  \ldots, q_n$, i.e., each of these state can be reached on $u$.  Now,
  if the union of these states' individual languages corresponds to the
  universal language, $\Sigma^\omega$, then clearly $u \in good(\cA)$.

  ($\supseteq$):
  Take any $u \in good(\cA)$.  By the previous proposition and the
  construction of $\cG^p_\cA$, there exist runs on $u$ in $\cG^p_\cA$ to
  states $q_0, \ldots, q_n$, each of which is accepting, i.e., there is
  at least one such run.  Moreover from the construction of $\cG_\cA$,
  whose state graph corresponds to the deterministic variant of
  $\cG^p_\cA$, it follows that there has to be a state $\langle q_0,
  \ldots, q_n \rangle \in Q'$ which can be reached on $u$.  We now have
  to show that $\langle q_0, \ldots, q_n \rangle$ is an accepting state
  of $\cG_\cA$.  For assume not, i.e., $\cL(\cA(q_0)) \cup \ldots \cup
  \cL(\cA(q_n)) \neq \Sigma^\omega$ holds, then there exists a word $w
  \in \Sigma^\omega$, 
  such that $w \not\in \cL(\cA(q_0)) \cup \ldots \cup \cL(\cA(q_n))$,
  and consequently $uw \not\in \cL(\cA)$.  Clearly, then $u \not\in
  good(\cL(\cA))$.  Contradiction.
  \qed
\end{proof}

\begin{remark}
  \label{rem:tight}
  In \cite{DBLP:conf/atva/KupfermanL06}, Kupferman and Lampert discuss
  properties of an NFA, referred to as a ``tight automaton,'' that
  accepts all the good prefixes of some NBA, $\cA$.
  The name stems from the fact that their paper is more concerned with
  the construction of so called ``fine automata,'' which accept only
  \emph{some} good prefixes, but not all. 
  %
  Although they do not explicitly give details on how to obtain a tight
  automaton, and only consider the special case where the NBA describes
  a co-safety language, they conclude, using a language-theoretic
  argument, that such an automaton must, in the worst-case, be of
  exponential size wrt.\ $\cA$---which agrees with our procedure above.
  %
  Their restriction to only examine NBAs which describe co-safety
  languages seems motivated solely by their application of model
  checking (co-) safety languages.  Consequently, they discuss how to
  obtain tight automata for NBAs describing safety languages, then
  accepting all the bad prefixes, and tight automata for for NBAs
  describing co-safety languages, then accepting all the good prefixes.
  However, it is easy to see that the constructions outlined on an
  abstract level by Kupferman and Lampert easily transfer to general
  NBAs, and result in the above described procedure when an acceptor for
  good prefixes is needed.  Hence, $\cG_\cA$ can be considered as a
  general form of a tight automaton capturing good prefixes, regardless
  as to whether $\cA$ describes a co-safety language, or not.
\end{remark}


\subsection{Deciding monitorability---The general case}

Now that we have all the required tools at hand, let us continue to
prove this section's main result, namely the complexity of the
monitorability problem of \buchi automata.  We will do this by way of
the following lemmas, which provide sufficient and necessary conditions
for deciding the monitorability of a language defined by some NBA.

\begin{lemma}
  \label{lem:nba:2}
  If $\forall u \in good(safe(\cA)).\ \exists v \in \Sigma^\ast.\
  uv\Sigma^\omega \subseteq \cL(live(\cA))$, then $\cA$ is monitorable.
\end{lemma}

\begin{proof}
  For any $u \in good(safe(\cA))$ there does not exist a
  $v\in\Sigma^\ast$, such that $uv\Sigma^\omega \cap \cL(\cA) =
  \emptyset$ as $u$ is a good prefix of $\cL(safe(\cA))$, and
  $\cL(live(\cA))$ does not, by definition of $live(\cA)$, have any bad
  prefixes.  Now, if the assumption $\forall u \in
  good(safe(\cA)).\ \exists v \in \Sigma^\ast.\ uv\Sigma^\omega
  \subseteq \cL(live(\cA))$ holds, then any such $u$ is extensible to to
  be a good prefix of $\cL(live(\cA))$ and thus $\cL(\cA)$.

  On the other hand, if $u \not\in good(safe(\cA))$, then by the
  definition of a closed set, this $u$ is extensible to be a bad prefix
  of $\cL(safe(\cA))$ and thus $\cL(\cA)$.

  As for any $u \in \Sigma^\ast$, either $u \in good(safe(\cA))$, or
  not, any $u$ can be extended to be either a good or a bad prefix of
  $\cL(\cA)$ under the lemma's assumption.  \qed
\end{proof}

\begin{corollary}
  \label{lem:nba:1}
  If there exists no good prefix of $\cL(safe(\cA))$, then $\cA$ is
  monitorable.
\end{corollary}

\begin{lemma}
  \label{lem:nba:3}
  Let $\cA$ be monitorable, then 
  $\forall u \in good(safe(\cA)).\ \exists v \in \Sigma^\ast.\
  uv\Sigma^\omega \subseteq \cL(live(\cA))$.
\end{lemma}

\begin{proof}
  We are going to show the contrapositive of the lemma's statement; that
  is, $\exists u \in good(safe(\cA)).\ \forall v \in \Sigma^\ast.\ uv
  \not\in good(live(\cA))$ implies that $\cA$ is \emph{not} monitorable.
  Let us now fix such a prefix $u$.
  Since $u \in good(safe(\cA))$, for $\cA$ to be monitorable after $u$,
  there would have to exist some $v \in \Sigma^\ast$, such that $uv \in
  good(live(\cA))$, thus $uv \in good(safe(\cA) \cap live(\cA))$, and
  therefore $uv \in good(\cA)$.
  However, by assumption this is not possible.
  Hence, $u$ is an ``ugly prefix'' of $\cL(\cA)$ and, consequently,
  $\cA$ not monitorable.
  \qed
\end{proof}


\begin{theorem}
  The monitorability problem of \buchi automata is decidable in PSpace.
\end{theorem}

\begin{proof}
  Observe that due to Lemma~\ref{lem:nba:2} and \ref{lem:nba:3}, the
  monitorability of $\cA$ is decidable in PSpace if and only if it can
  be checked in PSpace, whether the following is true:
  \begin{equation}
    \forall u \in
    good(safe(\cA)).\ \exists v \in \Sigma^\ast.\ uv\Sigma^\omega
    \subseteq \cL(live(\cA)).
  \end{equation}
  However, instead of giving an algorithm for checking if, for some
  $\cA$, this property holds, we devise an algorithm that returns $true$
  if the complementary statement holds, i.e., if
  \[
  \exists u \in good(safe(\cA)).\ \forall v \in \Sigma^\ast.\
  uv\Sigma^\omega \not\subseteq \cL(live(\cA))
  \]
  is true.
  For some $\cA$ this is the case if there exists some finite word $u
  \in good(sa\-fe(\cA))$, such that $u$ cannot be extended to be a good
  prefix of $\cL(live(\cA))$.
  Let, therefore, $\cG$ be the tight automaton over $\Sigma$, such that
  $\cL(\cG) = good(safe(\cA))$.
  As pointed out in Remark~\ref{rem:tight}, $\cG$ may, in the worst-case
  be of exponential size wrt.\ $safe(\cA)$, which stems from the fact
  that a standard subset construction needs to be applied.  In our case
  this means that the states of $\cG$ are the exponentially many sets of
  states of $safe(\cA)$.
  %
  %
  Therefore, we can only guess, in a step-wise manner, a path through
  $\cG \times live(\cA)$, corresponding to a word $u\in\Sigma^\ast$, to
  a pair of states $(q, q')$, where $q$ is now a set of states of
  $safe(\cA)$.  Note that using $\delta$ of $\cA$, we can easily check
  the connectedness of two states in $\cG \times live(\cA)$ in PSpace.
  Next, we check if $q$ is an accepting state in $\cG$, which, by
  Proposition~\ref{prop:nba:tight}, is the case if and only if
  the states $q_0, \ldots, q_n \in q$ are such that $\cL(safe(\cA)(q_0))
  \cup \ldots \cup \cL(safe(\cA)(q_n)) = \Sigma^\omega$.
  %
  %
  This property can be checked in PSpace in the size of $safe(\cA)$,
  because the union of two NBAs is of polynomial size and determining
  language equivalence of two NBAs is a PSpace-complete problem
  \cite{DBLP:conf/icalp/SistlaVW85}.
  Moreover, as $q'$ was reached on $u$,
  %
  %
  we have $uv\Sigma^\omega \not\subseteq \cL(live(\cA))$ for all
  possible extensions $v \in \Sigma^\ast$ if and only if $live(\cA)$
  does not contain a state $p$ that is reachable from $q'$, such that
  $\cL(live(\cA)(p))$ is open.
  Using the algorithm presented in \cite{bauer:pinchinat:deds09}, which
  we have employed before, this can be checked in PSpace as well.
  So, if no such $p$ exists and $q$ is an accepting state in $\cG$, then
  obviously $uv\Sigma^\omega \not\subseteq \cL(live(\cA))$ for some $u$
  and all its possible extensions $v \in \Sigma^\ast$ and, consequently,
  our algorithm returns $true$.
  
  This procedure for checking if the complement of $(2)$
  holds 
  is nondeterministic and does not use more than polynomial space wrt.\
  the size of $\cA$, and hence is in NPSpace.  Again, as NPSpace =
  PSpace = co-PSpace, the statement follows.
  %
  \qed
\end{proof}

\begin{theorem}
  \label{thm:nba:hardness}
  The monitorability problem of \buchi automata is PSpace-complete.
\end{theorem}

\begin{proof}
  It is sufficient to show PSpace-hardness.
  We proceed by reducing the PSpace-complete problem of checking if some
  NFA $\cB$ over some alphabet $\Sigma = \{ a_1, \ldots, a_n \}$, is
  such that $\cL(\cB) = \Sigma^\ast$ \cite{578533}.  In other words, we
  construct for $\cB$ in at most polynomial time an NBA, $\cA$, such
  that $\cA$ is monitorable if and only if $\cL(\cB) = \Sigma^\ast$.

  Let us first check if $\cL(\cB) = \emptyset$ is true.  It is well
  known that this can be done in polynomial time (cf.\ \cite{578533}).
  If the answer is ``yes'', we return the NBA, $\cA$, which corresponds
  to the models of the LTL formula $\ltlG\ltlF a$ over the alphabet
  $\Sigma' := \{a, b\}$, which by Proposition~\ref{prop:live} is
  non-monitorable.

  If $\cL(\cB) \neq \emptyset$, we proceed as follows.  Let $\Sigma_1 :=
  \{ a^1_1, \ldots, a^1_n \}$, $\Sigma_2 := \{ a^2_1, \ldots, a^2_n
  \}$ be alphabets.  Let us construct, in linear time, an NFA, $\cB_1$,
  respectively $\cB_2$, which is like $\cB$, except that it accepts
  $\cL(\cB)$ projected onto $\Sigma_1$, respectively onto $\Sigma_2$.
  Let $B_1$, respectively $B_2$, be the language accepted by $\cB_1$,
  respectively $\cB_2$.
  We now construct an NBA $\cA$, such that it accepts the following
  language, split into three parts for readability:
  \[
  \begin{array}{rll}
    (i) & & ((\Sigma_1 \cup \Sigma_2)^\ast (B_1B_2 \cup B_2B_1))^\omega\\
    (ii) & & \cup\ ((\Sigma_1 \cup \Sigma_2)^\ast B_1)^\omega\\
    (iii) & & \cup\ ((\Sigma_1 \cup \Sigma_2)^\ast B_2)^\omega.\\
  \end{array}
  \]
  It is easy to see that $\cA$ can be constructed in time no more than
  polynomial wrt.\ the size of $\cB$.  We now prove the following two
  claims.

  Let $\cL(\cB) = \Sigma^\ast$, then $\cA$ is monitorable:
  Notice first that a word $w \in (\Sigma_1 \cup \Sigma_2)^\omega$
  either is
  \begin{itemize}
  \item an alternation of finite words over $\Sigma_1$ and $\Sigma_2$,
  \item entirely over $\Sigma_1$ (respectively, $\Sigma_2$),
  \item an alternation of finite words over $\Sigma_1$ and $\Sigma_2$,
    followed by an infinite word over $\Sigma_1$ or $\Sigma_2$.
  \end{itemize}
  One can easily verify that all these cases are covered by the language
  accepted by $\cA$.  Hence, if $\cL(\cB) = \Sigma^\ast$, then $\cL(\cA)
  = \Sigma^\omega$, and therefore $\cA$ is monitorable.

  Let $\cA$ be monitorable, then $\cL(\cB) = \Sigma^\ast$: For assume not,
  that is, we assume $\cA$ is monitorable, but that $\cL(\cB) \neq
  \Sigma^\ast$ holds.
  From the latter it follows that there must exist a finite word $u \in
  \Sigma^\ast$, corresponding to some word $u' \in \Sigma_1^\ast$, such
  that $u \not\in \cL(\cB)$, and consequently $u' \not\in B_1$
  (respectively, for $B_2$).
  Due to way we have chosen the $\omega$-regular expression above, this
  $u'$ implies the existence of a language $L \subseteq (\Sigma_1 \cup
  \Sigma_2)^\omega$, such that $L \not\subseteq \cL(\cA)$; for example,
  we can easily prove that
  $L := (u'(\Sigma_1 \cup \Sigma_2))^\omega$ is not a subset of any of
  the three sets given by $(i)-(iii)$ above, and hence $L \not\subseteq
  \cL(\cA)$.
  Therefore, $\cA$ is not universal over $(\Sigma_1 \cup
  \Sigma_2)^\omega$.
  Notice further that all words $w \in \cL(\cA)$ are such that they
  require infinitely often the occurrence of a finite word $u$ either in
  $(i)$ $B_1B_2$ and interchangeably with $B_2B_1$, $(ii)$ $B_1$, or
  $(iii)$ $B_2$.
  More concretely, all infinite $w$ are such that finite words of the
  form
  \[
  \begin{array}{rll}
    (i) & & a^1_i\ldots a^1_j a^2_k\ldots a^2_l\hbox{ and optionally the mirrored
    version occur infinitely often, or}\\
    (ii) & & a^1_i\ldots a^1_j\hbox{ occurs infinitely often, or }\\
    (iii) & & a^2_i\ldots a^2_j\hbox{ occurs infinitely often,}
  \end{array}
  \]
  where, for all indices $g$, we have $a^1_g \in \Sigma_1$ and $a^2_g
  \in \Sigma_2$.
  In what follows, let $L_i$, $L_{ii}$, and $L_{iii}$ be the languages
  corresponding to the sets given by $(i)$, $(ii)$, and $(iii)$,
  respectively.
  It is obvious that $L_i$, $L_{ii}$ and $L_{iii}$ each define a dense
  but not open set over the words in $(\Sigma_1 \cup \Sigma_2)^\omega$
  as the infinite repetition of a finite word is required in each case.
  Moreover, as dense sets are closed under union \cite{93442}, and $L_i
  \cup L_{ii} \cup L_{iii} = \cL(\cA)$, it follows that $\cA$ defines a
  dense but not open set.  Together with the fact that $\cA$ is not
  universal, it follows that $\cA$ defines a classical liveness
  property, i.e., is not monitorable.  Contradiction.
  \qed
\end{proof}

\subsection{Deciding monitorability---The deterministic case}
\label{sec:nba:dba}

It is well-known that languages expressible by deterministic \buchi
automata (DBAs) are strictly less expressive than the ones accepted by
general (or, nondeterministic) NBAs: For example, one cannot express the
language given by the $\omega$-regular expression $(a + b)^*a^\omega$
over $\Sigma = \{a, b\}$ as can be easily proven.
%
%
%
On the other hand, it is possible to represent all safety and co-safety
languages using DBAs (cf.\ \cite{1217614}), although not every
DBA-representable language is necessarily monitorable as the example
over $\Sigma = \{a, b\}$, depicted in Fig.~\ref{fig:dba}, illustrates:
obviously, the language has neither good nor bad prefix.
Hence, it is reasonable wanting to be able to examine DBAs for their
monitorability as well.
\begin{figure}[tbhp]
  \centering
  \unitlength=4pt
  \begin{picture}(20,14)(0,-5)
    \gasset{Nw=5,Nh=5,Nadjustdist=2,curvedepth=0}
    \thinlines
    \node[Nmarks=ir,iangle=180](Q1)(0,0){$q_0$}
    \node[iangle=90,fangle=0](Q2)(20,0){$q_1$}
    \drawloop[loopangle=90](Q1){$a$}
    \drawloop[loopangle=90](Q2){$b$}
    \drawedge(Q1,Q2){$b$}
    \drawedge[curvedepth=4,ELside=l](Q2,Q1){$a$}
  \end{picture}
  \caption{Deterministic \buchi automaton over $\Sigma = \{a, b\}$
    describing a non-monitorable
    language.} 
  \label{fig:dba}
\end{figure}
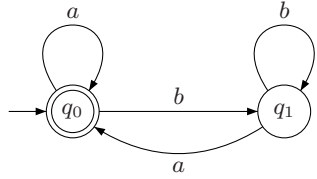
Not surprisingly though, if we know that the automaton in question is
deterministic, we can check its monitorability more efficiently than
before, using the criterion defind in Lemma~\ref{lem:nba:dba}.

However, before examining this condition, let us first make the
following assumption without loss of generality: let $\cA$ be a
\emph{complete} automaton; that is, for each symbol $a \in \Sigma$ and
each state $q \in Q$, there exists a state $q' \in Q$, such that
$\delta(q, a) = q'$.  It is easy to see that completing a deterministic
automaton takes time linear in the size of the automaton: one merely has
to add a ``trap''-state, the corresponding transitions, and self-loops
to it as necessary.  Let us use the symbol $\dagger$ to denote this
special state.
Moreover as a further notational convention, if there exists a path in
$\cA$, i.e., a state-action sequence, from state $q$ to $q'$, we also
write $q \leadsto q'$, or $q \leadsto_u q'$ to denote the fact that the
sequence of actions in this path corresponds to the finite word $u$.

\begin{lemma}
  \label{lem:nba:dba}
  A deterministic \buchi automaton $\cA$, defined as expected, is
  monitorable if and only if for every state $q \in Q$, it holds that
  \begin{itemize}
  \item a path exists such that $q \leadsto \dagger$, or
  \item a path exists, $q \leadsto q'$, with $\cL(\cA(q')) =
    \Sigma^\omega$.
  \end{itemize}
\end{lemma}

\begin{proof}
  As $\cA$ is deterministic, let, in what follows, $Q_0 = \{ q_0 \}$.

  If $q \leadsto \dagger$ holds then there exists a prefix $uv \in
  \Sigma^\ast$, such that $uv \in bad(\cA)$ with $q_0 \leadsto_u q$ and
  $q \leadsto_v \dagger$.
  On the other hand, if a path exists, $q \leadsto q'$, such that
  $\cL(\cA(q')) = \Sigma^\omega$, then there exists a prefix $uv \in
  \Sigma^\ast$, such that $uv \in good(\cA)$ with $q_0 \leadsto_u q$ and
  $q \leadsto_v q'$.
  Obviously, if every state implies the existence of either a bad or a
  good prefix, then $\cA$ is monitorable.

  For the other direction, assume the opposite, i.e., that $\cA$ is
  monitorable and that there exists a state $q$, such that
  there do not exist paths $(i)$ $q \leadsto \dagger$ and $(ii)$ $q
  \leadsto q'$, where $\cL(\cA(q')) = \Sigma^\omega$.
  Let $u \in \Sigma^\ast$ be the word defined by $q_0 \leadsto_u q$.
  From $(i)$ it follows that $u$ cannot be extended to be a bad prefix
  for $\cA$.  From $(ii)$ it follows that $u$ cannot be extended to be a
  good prefix for $\cA$.  Hence, $u$ is an ``ugly prefix'', and $\cA$
  not monitorable.  Contradiction.
  \qed
\end{proof}

\begin{theorem}
  \label{thm:dba}
  The monitorability problem of \buchi automata, when the automata are
  deterministic, can be solved in polynomial time.
\end{theorem}

\begin{proof}
  Recall, completion of $\cA$ takes linear time wrt.\ the size of $\cA$.
  So, without loss of generality, we assume the input automaton $\cA$
  complete already.
  Checking the condition of Lemma~\ref{lem:nba:dba} for $\cA$ means
  iterating through the $|Q|$ states of $\cA$ and checking for each $q
  \in Q$ whether any of the two sub-conditions holds.
  
  Using depth-first search, it is easy to see that the first condition
  can be checked in polynomial time (in fact, in time $O(|Q| + |Q| \cdot
  |\Sigma|) = O(|Q| \cdot |\Sigma|)$ as there are $|Q| \cdot |\Sigma|$
  transitions in a complete automaton).

  The second condition involves checking for each reachable state, $q'$,
  from state $q$, whether or not $\cL(\cA(q')) = \Sigma^\omega$.
  In the general case, i.e., when $\cA$ is nondeterministic, the latter
  problem is known to be PSpace-complete in the size of $\cA$.  However,
  as $\cA$ is deterministic, this condition can, in fact, be checked in
  time linear wrt.\ the size of $\cA$:
  In \cite{Kurshan198759}, Kurshan outlines a construction for a DBA
  $\cA'$, such that $\cL(\cA') = \overline{\cL(\cA)}$, where $\cA'$ has
  only $2|Q|$ states.  Now, checking if $\cL(\cA') = \emptyset$ holds is
  known to be LogSpace-complete for NLogSpace \cite{191907} and
  clearly the case if and only if $\cL(\cA) = \Sigma^\omega$ holds.

  From these two observations it now easily follows that checking both
  conditions of Lemma~\ref{lem:nba:dba} can be done in no more than
  polynomial time wrt.\ the size of $\cA$.
  \qed
\end{proof}

Finally, observe that the non-monitorable DBA depicted in
Fig.~\ref{fig:dba} is complete for $\Sigma = \{a, b\}$, but 
\emph{incomplete} and \emph{monitorable} for $\Sigma = \{ a, b, c \}$.


\section{Closure of the monitorable $\omega$-languages}
\label{sec:closure}
We now examine closure properties of monitorable $\omega$-languages.
Let us fix two languages $L,M\subseteq \Sigma^\omega$ for the remainder
of this section.

\begin{proposition}
  Let $L$ and $M$ be monitorable, then $L \cap M$ is monitorable.
\end{proposition}

\begin{proof}
  Since $L$ is monitorable two cases arise: every $u \in \Sigma^\ast$ is
  extensible to be a good prefix of $L$ or to a bad prefix of $L$ (or
  both, but this case is covered in the following):
  
  $(i)$ Let us fix some $u \in \Sigma^\ast$, such that $u\Sigma^\omega
  \cap L = \emptyset$ holds.  Then, irrespective of $M$, $u\Sigma^\omega
  \cap L \cap M = \emptyset$, and hence $u \in bad(L \cap M)$.
  $(ii)$ Let us fix some $u \in \Sigma^\ast$, such that $u\Sigma^\omega
  \subseteq L$ holds.
  Then, by the monitorability of $M$,
  \[
  \exists v \in \Sigma^\ast.\ uv\Sigma^\omega \cap M = \emptyset \vee
  uv\Sigma^\omega \subseteq M.
  \]
  As before, if $uv\Sigma^\omega \cap M = \emptyset$, then
  $uv\Sigma^\omega \cap M \cap L = \emptyset$ and hence $uv \in bad(L
  \cap M)$.  On the other hand, if $uv\Sigma^\omega \subseteq M$, then
  $uv\Sigma^\omega \subseteq L \cap M$ and, consequently, $uv \in good(L
  \cap M)$.
 
  As all finite words are extensible to be either good or bad prefixes
  of $L$, and in either case it is possible to find a good or a bad
  prefix of $L \cap M$, we conclude that $L \cap M$ is monitorable as
  well
  \qed
\end{proof}

\begin{proposition}
  Let $L$ be monitorable, then $\overline{L}$ is monitorable.
\end{proposition}

\begin{proof}
  Follows directly from applying Proposition~\ref{prop:prefixes} to
  Proposition~\ref{prop:mon}.
  \qed
\end{proof}

\begin{theorem}
  \label{thm:closure}
  The monitorable $\omega$-languages are closed under (finitary)
  application of intersection, complement, and union.
\end{theorem}

\begin{proof}
  Follows now easily from the fact that $L \cup M =
  \overline{\overline{L} \cap \overline{M}}$.
  \qed
\end{proof}



\section{Conclusions}
\label{sec:end}

\label{sec:conc}

The formal concept of monitorability of an $\omega$-regular language was
first introduced by Pnueli and Zaks \cite{DBLP:conf/fm/PnueliZ06}.
%
A subsequent result of \cite{tum-i0724} implies that the monitorability
problem of LTL and NBAs as laid out in Sec.~\ref{sec:mon} is, in fact,
decidable using a 2ExpSpace algorithm, whereas
\cite{DBLP:conf/rv/FalconeFM09} recently could show that the monitorable
$\omega$-languages are strictly more expressive than the commonly used
set of safety properties (and, in fact, an unrestricted Boolean
combination thereof), known to be PSpace-complete when the language is
given by an LTL formula or an NBA.
The present paper closes the exponential ``gap'' that lies between these
observations, in that it shows that the monitorability problem of LTL
and NBAs are both, in fact, PSpace-complete (unless, of course, the NBAs
are, in fact, deterministic).
%

Besides being of theoretical merit in order to being able to classify
the monitorable $\omega$-regular languages wrt.\ existing
classifications such as the safety-progress hierarchy, a practical
interpretation of this result is that checking the monitorability of a
formal specification, given as LTL formula or by an NBA, is
computationally as involved as checking, say, if the specification
defines a safety language.
Moreover, knowing the upper bound of the problem, we can devise new and
probably faster algorithms than \cite{tum-i0724} for checking the
monitorability of specifications, such that users can determine---prior
to the actual runtime verification process, or any attempts to build a
monitor---whether or not their specifications are monitorable at all.

As a final theoretical contribution, it is worth pointing out that,
using the results of this paper, one can easily show that the
monitorability problem of $\omega$-regular expressions is also
PSpace-complete, since there exists a polynomial time transformation
from $\omega$-regular expressions to NBAs, which yields membership in
PSpace.  Together with the proof of Theorem~\ref{thm:nba:hardness},
where we used $\omega$-regular expressions, we then obtain completeness
for this problem.




\paragraph{Acknowledgements.}
%
%
Alban Grastien helped remove a bug in an earlier proof of what later
became Lemma~\ref{lem:nba:2}.
Many hours of fruitful discussion with Jussi Rintanen on the upper
bounds of the monitorability problem of LTL are gratefully acknowledged.




\bibliographystyle{plain}
\bibliography{shortstrings,bibliography}

\begin{thebibliography}{10}

\bibitem{alpern87recognizing}
Bowen Alpern and Fred~B. Schneider.
\newblock Recognizing safety and liveness.
\newblock {\em Distributed Computing}, 2(3):117--126, 1987.

\bibitem{tum-i0724}
Andreas Bauer, Martin Leucker, and Christian Schallhart.
\newblock Runtime verification for {LTL} and {TLTL}.
\newblock Technical Report TUM-I0724, Technische Universit\"at M\"unchen,
  December 2007.
\newblock An extended version to appear in \emph{ACM Transactions on Software
  Engineering and Methodology}.

\bibitem{bauer:pinchinat:deds09}
Andreas Bauer and Sophie Pinchinat.
\newblock Prognosis of $\omega$-languages for the diagnosis of *-lang\-uages: A
  topological perspective.
\newblock {\em Discrete Event Dynamic Systems}, 19(4):451--470, 2009.

\bibitem{DBLP:conf/cav/dAmorimR05}
Marcelo d'Amorim and Grigore Ro\c{s}u.
\newblock Efficient monitoring of $\omega$-languages.
\newblock In K.~Etessami and S.~K. Rajamani, editors, {\em Proc.\ 17th
  International Conference on Computer Aided Verification (CAV)}, volume 3576
  of {\em LNCS}, pages 364--378. Springer, 2005.

\bibitem{verimag-TR-2009-6}
Yli{\`e}s Falcone, Jean-Claude Fernandez, and Laurent Mounier.
\newblock Runtime verification of safety-progress properties.
\newblock Technical report, VERIMAG, Centre \'Equation, Gi\`eres, June 2009.

\bibitem{DBLP:conf/rv/FalconeFM09}
Yli{\`e}s Falcone, Jean-Claude Fernandez, and Laurent Mounier.
\newblock Runtime verification of safety-progress properties.
\newblock In S.~Bensalem and D.~Peled, editors, {\em Proc.\ 9th International
  Workshop on Runtime Verification (RV)}, volume 5779 of {\em LNCS}, pages
  40--59. Springer, 2009.

\bibitem{578533}
Michael~R. Garey and David~S. Johnson.
\newblock {\em Computers and Intractability: A Guide to the Theory of
  NP-Completeness}.
\newblock W. H. Freeman \& Co., 1979.

\bibitem{havelund02synthesizing}
Klaus Havelund and Grigore Ro\c{s}u.
\newblock {Synthesizing Monitors for Safety Properties}.
\newblock In J.-P. Katoen and P.~Stevens, editors, {\em Proc.\ 8th
  International Conference on Tools and Algorithms for the Construction and
  Analysis of Systems (TACAS)}, volume 2280 of {\em LNCS}, pages 342--356.
  Springer, 2002.

\bibitem{havelund2004}
Klaus Havelund and Grigore Ro\c{s}u.
\newblock Efficient monitoring of safety properties.
\newblock {\em Software Tools for Technology Transfer (STTT)}, 6(2):158--173,
  2004.

\bibitem{1217614}
Joachim Klein and Christel Baier.
\newblock Experiments with deterministic {$\omega$}-automata for formulas of
  linear temporal logic.
\newblock {\em Theoretical Computer Science}, 363(2):182--195, 2006.

\bibitem{DBLP:conf/atva/KupfermanL06}
Orna Kupferman and Robby Lampert.
\newblock On the construction of fine automata for safety properties.
\newblock In S.~Graf and W.~Zhang, editors, {\em Proc.\ 4th International
  Symposium on Automated Technology for Verification and Analysis (ATVA)},
  volume 4218 of {\em LNCS}, pages 110--124. Springer, 2006.

\bibitem{Kurshan198759}
Robert~P. Kurshan.
\newblock Complementing deterministic {B}{\"u}chi automata in polynomial time.
\newblock {\em Computer and System Sciences}, 35(1):59--71, 1987.

\bibitem{93442}
Zohar Manna and Amir Pnueli.
\newblock A hierarchy of temporal properties.
\newblock In {\em Proc.\ 9th ACM Symposium on Principles of Distributed
  Computing}, pages 377--410. ACM, 1990.

\bibitem{cj97-01-15}
Christos Papadimitriou.
\newblock {\em Computational Complexity}.
\newblock Addison-Wesley, 1995.

\bibitem{Pnueli77}
Amir Pnueli.
\newblock The temporal logic of programs.
\newblock In {\em Proc. 18th IEEE Symposium on the Foundations of Computer
  Science (FOCS)}, pages 46--57. IEEE, 1977.

\bibitem{DBLP:conf/fm/PnueliZ06}
Amir Pnueli and Aleksandr Zaks.
\newblock {PSL} model checking and run-time verification via testers.
\newblock In J.~Misra, T.~Nipkow, and E.~Sekerinski, editors, {\em
  International Symposium on Formal Methods (FM)}, volume 4085 of {\em LNCS},
  pages 573--586. Springer, 2006.

\bibitem{rosu-havelund-2005-jase}
Grigore Ro\c{s}u and Klaus Havelund.
\newblock Rewriting-based techniques for runtime verification.
\newblock {\em Automated Software Engineering}, 12(2):151--197, 2005.

\bibitem{DBLP:journals/tissec/Schneider00}
Fred~B. Schneider.
\newblock Enforceable security policies.
\newblock {\em ACM Transactions on Information and System Security},
  3(1):30--50, 2000.

\bibitem{DBLP:journals/fac/Sistla94}
A.~Prasad Sistla.
\newblock Safety, liveness and fairness in temporal logic.
\newblock {\em Formal Aspects of Computing}, 6(5):495--512, 1994.

\bibitem{DBLP:journals/jacm/SistlaC85}
A.~Prasad Sistla and Edmund~M. Clarke.
\newblock The complexity of propositional linear temporal logics.
\newblock {\em Journal of the ACM}, 32(3):733--749, 1985.

\bibitem{DBLP:conf/icalp/SistlaVW85}
A.~Prasad Sistla, Moshe~Y. Vardi, and Pierre Wolper.
\newblock The complementation problem for {B}{\"u}chi automata with
  applications to temporal logic (extended abstract).
\newblock In W.~Brauer, editor, {\em Proc.\ 12th Colloquium on Automata,
  Languages and Programming (ICALP)}, volume 194 of {\em LNCS}, pages 465--474.
  Springer, 1985.

\bibitem{114895}
Wolfgang Thomas.
\newblock Automata on infinite objects.
\newblock In {\em Handbook of theoretical computer science (vol. B): formal
  models and semantics}, pages 133--191. MIT Press, 1990.

\bibitem{DBLP:journals/corr/cs-LO-0101017}
Ulrich Ultes-Nitsche and Pierre Wolper.
\newblock Checking properties within fairness and behavior abstractions.
\newblock {\em CoRR}, cs.LO/0101017, 2001.
\newblock A shorter version of this article has also appeared in \emph{Proc.\
  16th {ACM} Symposium on Principles of Distributed Computing}, 1997.

\bibitem{191907}
Moshe~Y. Vardi and Pierre Wolper.
\newblock Reasoning about infinite computations.
\newblock {\em Information and Computation}, 115(1):1--37, 1994.

\end{thebibliography}

\end{document}